\documentclass[pra,superscriptaddress,nopacs,twocolumn,longbibliography]{revtex4-1}

\usepackage{xcolor}
\definecolor{darkred}{rgb}{0.8,0.1,0.1}

\usepackage{soul}
\usepackage{amsmath,mathtools}
\usepackage{enumerate}
\usepackage{graphicx,epsfig,amsmath,amssymb,verbatim,color}
\usepackage{dsfont}
\usepackage{float}
\usepackage{tikz}

\usepackage[T1]{fontenc}
\usepackage[utf8]{inputenc}

\usepackage{hyperref}
\hypersetup{colorlinks=true,citecolor=blue,linkcolor=blue,filecolor=blue,urlcolor=blue,breaklinks=true}
\usepackage{cleveref}
\usepackage{graphicx}
\usepackage{xcolor}

\definecolor{darkblue}{RGB}{0,76,156}
\definecolor{darkkblue}{RGB}{0,0,153}
\definecolor{blue2}{RGB}{102,178,255}

\usepackage{url}
\usepackage{amsthm}

\newtheorem{proposition}{Proposition}
\newtheorem{lemma}[proposition]{Lemma}

\newtheorem{theorem}[proposition]{Theorem}
\newtheorem{remark}{Remark}
\newtheorem{corollary}[proposition]{Corollary}

\def\squareforqed{\hbox{\rlap{$\sqcap$}$\sqcup$}}
\def\qed{\ifmmode\squareforqed\else{\unskip\nobreak\hfil
\penalty50\hskip1em\null\nobreak\hfil\squareforqed
\parfillskip=0pt\finalhyphendemerits=0\endgraf}\fi}
\def\endenv{\ifmmode\;\else{\unskip\nobreak\hfil
\penalty50\hskip1em\null\nobreak\hfil\;
\parfillskip=0pt\finalhyphendemerits=0\endgraf}\fi}

\mathchardef\ordinarycolon\mathcode`\:
\mathcode`\:=\string"8000
\def\vcentcolon{\mathrel{\mathop\ordinarycolon}}
\begingroup \catcode`\:=\active
  \lowercase{\endgroup
  \let :\vcentcolon
  }

\makeatletter
\def\resetMathstrut@{%
    \setbox\z@\hbox{%
        \mathchardef\@tempa\mathcode`\[\relax
        \def\@tempb##1"##2##3{\the\textfont"##3\char"}%
        \expandafter\@tempb\meaning\@tempa \relax
    }%
    \ht\Mathstrutbox@\ht\z@ \dp\Mathstrutbox@\dp\z@}

\newcommand{\nc}{\newcommand}
\nc{\rnc}{\renewcommand}
\nc{\beg}{\begin{equation}}
\nc{\eeq}{{\end{equation}}}
\nc{\beqa}{\begin{eqnarray}}
\nc{\eeqa}{\end{eqnarray}}
\nc{\lbar}[1]{\overline{#1}}
\nc{\bra}[1]{\langle#1|}
\nc{\ket}[1]{|#1\rangle}
\nc{\ketbra}[2]{|#1\rangle\!\langle#2|}
\nc{\braket}[2]{\langle#1|#2\rangle}

\nc{\proj}[1]{| #1\rangle\!\langle #1 |}
\nc{\avg}[1]{\langle#1\rangle}
\nc{\Rank}{\operatorname{Rank}}
\nc{\smfrac}[2]{\mbox{$\frac{#1}{#2}$}}
\nc{\tr}{\operatorname{Tr}}
\nc{\ox}{\otimes}
\nc{\dg}{\dagger}
\nc{\dn}{\downarrow}
\nc{\cA}{{\cal A}}
\nc{\cB}{{\cal B}}
\nc{\cC}{{\cal C}}
\nc{\cD}{{\cal D}}
\nc{\cE}{{\cal E}}
\nc{\cF}{{\cal F}}
\nc{\cG}{{\cal G}}
\nc{\cH}{{\cal H}}
\nc{\cI}{{\cal I}}
\nc{\cJ}{{\cal J}}
\nc{\cK}{{\cal K}}
\nc{\cL}{{\cal L}}
\nc{\cM}{{\cal M}}
\nc{\cN}{{\cal N}}
\nc{\cO}{{\cal O}}
\nc{\cP}{{\cal P}}
\nc{\cQ}{{\cal Q}}
\nc{\cR}{{\cal R}}
\nc{\cS}{{\cal S}}
\nc{\cT}{{\cal T}}
\nc{\cV}{{\cal V}}
\nc{\cU}{{\cal U}}
\nc{\cX}{{\cal X}}
\nc{\cY}{{\cal Y}}
\nc{\cZ}{{\cal Z}}
\nc{\cW}{{\cal W}}
\nc{\csupp}{{\operatorname{csupp}}}
\nc{\qsupp}{{\operatorname{qsupp}}}
\nc{\var}{{\operatorname{var}}}
\nc{\rar}{\rightarrow}
\nc{\lrar}{\longrightarrow}
\nc{\polylog}{{\operatorname{polylog}}}
\nc{\1}{{\mathds{1}}}
\nc{\wt}{{\operatorname{wt}}}
\nc{\av}[1]{{\left\langle {#1} \right\rangle}}
\nc{\supp}{{\operatorname{supp}}}

\nc{\dia}{{\diamondsuit }}

\def\ve{\varepsilon}

\def\x{\xi}

\def\o{\omega}

\nc{\RR}{{{\mathbb R}}}
\nc{\CC}{{{\mathbb C}}}
\nc{\FF}{{{\mathbb F}}}
\nc{\NN}{{{\mathbb N}}}
\nc{\ZZ}{{{\mathbb Z}}}
\nc{\PP}{{{\mathbb P}}}
\nc{\QQ}{{{\mathbb Q}}}
\nc{\UU}{{{\mathbb U}}}
\nc{\EE}{{{\mathbb E}}}
\nc{\id}{{\operatorname{id}}}

\nc{\CHSH}{{\operatorname{CHSH}}}

\nc{\be}{\begin{equation}}
\nc{\ee}{{\end{equation}}}
\nc{\bea}{\begin{eqnarray}}
\nc{\eea}{\end{eqnarray}}
\nc{\<}{\langle}
\rnc{\>}{\rangle}
\nc{\Hom}[2]{\mbox{Hom}(\CC^{#1},\CC^{#2})}
\nc{\rU}{\mbox{U}}

\nc{\ob}[1]{#1}

\nc{\SEP}{{\text{SEP}}}
\nc{\NS}{{\text{NS}}}
\nc{\LOCC}{{\text{LOCC}}}
\nc{\PPT}{{\text{PPT}}}
\nc{\EXT}{{\text{EXT}}}
\nc{\Sym}{{\operatorname{Sym}}}

\nc{\HH}{\mathbb{H}}

\nc{\ERLO}{{E_{\text{r,LO}}}}
\nc{\ERLOCC}{{E_{\text{r,LOCC}}}}
\nc{\ERPPT}{{E_{\text{r,PPT}}}}
\nc{\ERLOCCinfty}{{E^{\infty}_{\text{r,LOCC}}}}
\nc{\Aram}{{\operatorname{\sf A}}}

\rnc{\bar}{\;\rule{0pt}{9.5pt}\right|\;}

\nc{\lset}{\left\{\left.}
\nc{\rset}{\right\}}

\nc{\lsetr}{\left\{}
\nc{\rsetr}{\right.\right\}}
\nc{\barr}{\left|\rule{0pt}{9.5pt}\;}

\let\id\1

\nc{\norm}[2]{\left\lVert#1\right\rVert_{#2\!}}
\nc{\lnorm}[2]{\left\lVert#1\right\rVert_{\ell_{#2}}}

\usepackage{pifont}
\nc{\sectag}[1]{\vspace{0.2cm} \noindent \textit{{#1---}}}
\nc{\exampletag}[1]{\vspace{0.2cm} \noindent \textit{{Example #1.}}}
\nc{\bcF}{\boldsymbol{\cF}}
\nc{\bcO}{\boldsymbol{\cO}}
\nc{\bcS}{\boldsymbol{\cS}}
\nc{\bTheta}{\boldsymbol{\Theta}}
\nc{\coneF}{\textsf{Cone}({\bcF})}
\nc{\sub}{\text{sub}}

\usepackage{pdfpages}
\makeatletter
\AtBeginDocument{\let\LS@rot\@undefined}
\makeatother

\usepackage{tikz}
\usepackage{hyperref}
\hypersetup{colorlinks=true,citecolor=blue,linkcolor=blue,filecolor=blue,urlcolor=blue,breaklinks=true}

\makeatletter
\def\grd@save@target#1{%
  \def\grd@target{#1}}
\def\grd@save@start#1{%
  \def\grd@start{#1}}
\tikzset{
  grid with coordinates/.style={
    to path={%
      \pgfextra{%
        \edef\grd@@target{(\tikztotarget)}%
        \tikz@scan@one@point\grd@save@target\grd@@target\relax
        \edef\grd@@start{(\tikztostart)}%
        \tikz@scan@one@point\grd@save@start\grd@@start\relax
        \draw[minor help lines,magenta] (\tikztostart) grid (\tikztotarget);
        \draw[major help lines] (\tikztostart) grid (\tikztotarget);
        \grd@start
        \pgfmathsetmacro{\grd@xa}{\the\pgf@x/1cm}
        \pgfmathsetmacro{\grd@ya}{\the\pgf@y/1cm}
        \grd@target
        \pgfmathsetmacro{\grd@xb}{\the\pgf@x/1cm}
        \pgfmathsetmacro{\grd@yb}{\the\pgf@y/1cm}
        \pgfmathsetmacro{\grd@xc}{\grd@xa + \pgfkeysvalueof{/tikz/grid with coordinates/major step}}
        \pgfmathsetmacro{\grd@yc}{\grd@ya + \pgfkeysvalueof{/tikz/grid with coordinates/major step}}
        \foreach \x in {\grd@xa,\grd@xc,...,\grd@xb}
        \node[anchor=north] at (\x,\grd@ya) {\pgfmathprintnumber{\x}};
        \foreach \y in {\grd@ya,\grd@yc,...,\grd@yb}
        \node[anchor=east] at (\grd@xa,\y) {\pgfmathprintnumber{\y}};
      }
    }
  },
  minor help lines/.style={
    help lines,
    step=\pgfkeysvalueof{/tikz/grid with coordinates/minor step}
  },
  major help lines/.style={
    help lines,
    line width=\pgfkeysvalueof{/tikz/grid with coordinates/major line width},
    step=\pgfkeysvalueof{/tikz/grid with coordinates/major step}
  },
  grid with coordinates/.cd,
  minor step/.initial=.2,
  major step/.initial=1,
  major line width/.initial=2pt,
}
\makeatother
\usetikzlibrary{shapes,arrows,spy,positioning,snakes}

\nc{\cam}{\affiliation{Department of Applied Mathematics and Theoretical Physics,\\ University of Cambridge, Cambridge, CB3 0WA, United Kingdom}}
\nc{\waterloo}{\affiliation{Institute for Quantum Computing, University of Waterloo, Waterloo, ON, N2L 3G1, Canada}}
\nc{\cadpi}{\affiliation{Perimeter Institute for Theoretical Physics, Waterloo, Ontario N2L 2Y5, Canada}}

\begin{document}
	\title{No-Go Theorems for Quantum Resource Purification}
	\author{Kun Fang}
	\email{kf383@cam.ac.uk}
	\cam \waterloo
	\author{Zi-Wen Liu}
	\email{zliu1@perimeterinstitute.ca}
	\cadpi
	
	\date{\today}

\begin{abstract}
The manipulation of quantum ``resources'' such as entanglement, coherence and magic states lies at the heart of quantum science and technology, empowering potential advantages over classical methods. In practice, a particularly important kind of manipulation is to ``purify'' the quantum resources, since they are inevitably contaminated by noise and thus often lose their power or become unreliable for direct usage. Here we prove fundamental limitations on how effectively generic noisy resources can be purified enforced by the laws of quantum mechanics, which universally apply to any reasonable kind of quantum resource. 
More explicitly, we derive nontrivial lower bounds on the error of converting any full-rank noisy state to any target pure resource state by any free protocol (including probabilistic ones)---it is impossible to achieve perfect resource purification, even probabilistically. Our theorems indicate strong limits on the efficiency of distillation, a widely used type of resource purification routine that underpins many key applications of quantum information science.  In particular, this general result induces the first explicit lower bounds on the resource cost of magic state distillation, a leading scheme for realizing scalable fault-tolerant quantum computation. Implications for the standard error-correction-based methods are specifically discussed. 
\end{abstract}

\maketitle

The field of quantum information takes a pragmatic approach to examining and utilizing quantum mechanics, seeking to obtain rigorous understandings of which information processing tasks can or cannot be accomplished according to the laws of nature. 
Efforts along this line since the 1980s have revolutionized our perception of physics and paved the way for many innovative technological applications such as quantum computation and communication \cite{dowling03,NielsenChuang}. 
In particular, the formulations of no-go (impossibility) theorems have played seminal roles---they often represent key advances in our understanding of quantum mechanics and have exerted profound influence on the development of quantum information science and technology.
A representative example is the no-cloning theorem~\cite{Wootters1982,Dieks1982}, which directly led to the invention of quantum error correction~\cite{Shor1995,Steane1996} and laid the foundation for plenty of other major quantum applications such as quantum cryptography~\cite{Bennett1984}, as well as advancing our understanding of the foundations of quantum mechanics \cite{Cloning_rmp,WoottersZurek}.

At the heart of the desired quantum information processing tasks is the manipulation of various useful quantum features, the most prominent examples being entanglement \cite{Horodecki:entanglement}, coherence \cite{coherenceRMP}, and ``magic'' \cite{bravyi2005,Veitch_2014}, that emerge as valuable ``resources'' that are needed to empower advantages over classical methods.  Such resource features can arise from all kinds of physical or conceptual restrictions on the feasible operations.  A prototypical example is the ``distant labs'' paradigm, where only local operations within the separate labs and classical communication between them (the so-called ``LOCC'') is allowed, rendering entanglement a resource that cannot be obtained for free and could, for instance, enable efficient quantum communication \cite{NielsenChuang,Horodecki:entanglement}. 

In practice, a particularly important and widely studied kind of manipulation is to ``purify'' the quantum resources, since quantum systems are highly susceptible to faulty controls and noise effects such as decoherence \cite{NielsenChuang,Preskill2018quantumcomputingin} that may jeopardize the power and reliability of quantum resources.   In particular, a standard procedure of quantum resource purification is to extract high-quality resource states better suited for application from a large amount of raw noisy 
ones, 
which is known as \emph{distillation}. 
Most notably, the distillation of entanglement \cite{BBPS96:ent_dist,BBPSSW96:ent_dist,BDSW96:ent_qec}, coherence \cite{WinterYang16,Fang2018,Regula2018} and magic states \cite{bravyi2005} has been extensively studied as key subroutines in quantum computation and communication. 
Therefore, understanding the limits to the efficiency of purification and distillation tasks is of great theoretical and practical importance.

To address this problem in a rigorous and general manner, we shall use the language of quantum resource theory (see \cite{ChitambarGour19} for an introduction of this framework), where each resource theory is defined by a set of \emph{free states} (in contrast to \emph{resource states}) and a set of \emph{free operations}.
Again take the entanglement theory as an example: the set of free states consists of the separable (unentangled) states, and LOCC is a standard choice of the set of free operations. Free states and operations can be flexibly defined, which gives rise to a wide variety of meaningful resource theories, as long as they follow a \emph{golden rule}: any free operation can only map a free state to another free state.  This simple rule selects the largest possible set of free operations allowed in resource manipulation, since any other operation can by definition create resources and thus trivialize the theory.   
Moreover, note that we are interested in the \emph{one-shot} setting as opposed to the conventional asymptotic setting here, since only a finite amount of resources is accessible in reality. We refer readers to Ref.~\cite{LBT19} for a general theory of the rates of one-shot resource manipulation.

In this work, we prove a set of no-go theorems for quantum resource purification that universally apply to any reasonable resource theory, manifesting that the production of any pure resource state with an arbitrarily small error, however weak this target state is, is generically prohibited by the golden rule. 
More formally, we establish quantitative bounds on the achievable accuracy of any free operation that is supposed to work with some probability.  
It turns out that there is a nontrivial trade-off between the accuracy and success probability, akin to the uncertainty relations.  
 The proofs follow from analyzing the peculiar properties of the hypothesis testing relative entropy monotone, a quantity known to characterize the efficiency of one-shot distillation in many cases \cite{BrandaoDatta11,fang2019distillation,Regula2018,ZLYCW19,YR16,WangWilde:magic,LBT19} but not studied in great depth.
Using the above results, we find lower bounds on the \emph{overhead} of distillation given by the number of copies of a certain primitive noisy state needed.  As a particularly important application, we derive specific lower bounds on the overhead of magic state distillation \cite{bravyi2005}, a leading proposal of fault-tolerant quantum computation \cite{NielsenChuang,Shor96,CampbellTerhalVuillot17}.  The consequent limitations to the common distillation schemes based on quantum error correction are discussed in relation to key advances in the search for better codes \cite{BravyiHaah12,HastingsHaah18,Haah2017magicstate,Haah2018codesprotocols}.  
Lastly, we provide a no-go theorem for the simulation of unitary resource channels, which is analogous to state purification, in accordance with the recent interest in extending conventional resource theory approaches for quantum states to quantum channels (see e.g.~\cite{Gour18comparison,Wilde18,PhysRevA.101.022335,TheurerEgloffZhangPlenio19,PhysRevX.9.031053,LiuYuan:channel,LiuWinter2018,GourWinter19} for general treatments).

We start by introducing the notations.
The sets of free operations and free states are respectively denoted by $\bcO$ and $\bcF$. 
They obey the golden rule that $\bcO \subseteq \widetilde\bcO$, where $\widetilde\bcO:=\{\cE \,|\, \forall \rho \in \bcF, \cE(\rho)\in \bcF\}$ (commonly known as the set of resource nongenerating operations in the literature).  
Note that virtually no assumptions on the specific properties of the resource theory are needed in this work, that is, $\bcF$ is almost completely up to one's choice, as long as there exists some resource pure state (technically, $\bcF$ is topologically closed and $\exists\,\psi\notin\bcF$) so that the purification task is well-defined. Even the convexity of $\bcF$, which is a common postulate for general resource theory results and frameworks (see e.g.~\cite{BrandaoGour15,LiuHuLloyd17:rdmap,TRBLA19:advantage,PhysRevA.101.022335,LBT19,PhysRevA.101.062315}), is not needed.

   The general goal of purification tasks is to transform some noisy primitive state to a pure target resource state by some protocol represented by a free operation.    In this work, we make a mild assumption that the density matrix representing the primitive state is full-rank, which holds generically for common noise effects and settings of practical interest such as multiple qubits.  We would also want to consider protocols that produce desired outputs with a certain probability, as long as we know when they do so (an important example being magic state distillation, as we shall discuss later).
To encompass such cases, consider the generalization of $\widetilde\bcO$ to the class $\widetilde\bcO_{\sub}:=\{\cL \,|\, \forall \rho \in \bcF, \exists\, t \geq 0, \sigma \in \bcF, \text{ s.t. } \cL(\rho) = t \cdot \sigma \}$, which consists of subnormalized quantum operations (sub-operations), i.e.~completely positive and trace-nonincreasing maps.
A free probabilistic protocol that transforms $\rho$ to $\gamma$ with probability $p$ and accuracy $1-\epsilon$ (or error $\epsilon$) is modeled by a quantum operation $\cE_{A\to XB}$ such that
$\cE_{A\to XB} (\rho_A) = \ket{0}\bra{0}_X\ox \cL_{A\to B}(\rho_A) + \ket{1}\bra{1}_X \ox \cG_{A\to B}(\rho_A)$. Here $X$ is an external flag register that keeps track of whether the protocol succeeds (0) or not (1); $\cL\in\bcO_{\sub}$ (any $\bcO_{\sub}\subseteq\widetilde{\bcO}_{\sub}$) is the free sub-operation representing the successful transformation such that $\cL_{A\to B}(\rho_A) = p \tau_B$ where $p = \mathrm{Tr}\cL(\rho)$ and $\tau$ is a density matrix satisfying $F(\tau,\gamma) \geq 1-\epsilon$ where $F(\rho,\sigma) := \|\sqrt{\rho}\sqrt{\sigma}\|_1^2$ is the fidelity between $\rho$ and $\sigma$.
The case where $\cL$ is a completely positive trace preserving (CPTP) map and thus $p=1$ corresponds to a deterministic protocol.

Now we are ready to introduce the explicit results.  
The following theorem reveals fundamental limitations on the accuracy and success probability of resource purification.
\begin{theorem}\label{thm: nogo 1}
Given any full-rank primitive state $\rho \not\in \bcF$ and any pure target resource state $\psi \not\in \bcF$, the following relation between the  success probability $p$ and transformation error $\epsilon$ must hold for any free probabilistic protocol:
\begin{equation}
\frac{\epsilon}{p} \geq  \frac{\lambda_{\min}(\rho)(1-f_\psi)}{1+R(\rho)}.
\end{equation}
where $\lambda_{\min}(\rho)$ is the smallest eigenvalue of $\rho$, $f_\psi:= \max_{\o\in \bcF} \tr (\psi \o)$ is the maximum overlap between $\psi$ and free states $\bcF$, and $R(\rho):= \min\{s| \exists\, s \geq 0, \text{state}~\sigma, \text{ s.t. } (\rho+ s\sigma)/(1+s) \in \bcF\}$ is the generalized robustness of state $\rho$.  For the deterministic case ($p=1$), the bound can be improved to $\epsilon \geq \lambda_{\min}(\rho)(1-f_\psi)$.
\end{theorem}
Notice that $f_\psi<1$ always holds by its definition, so the bound is always greater than zero, meaning that there is always a neighborhood of any $\psi$ that cannot be reached by any free protocol.   
This theorem establishes an ``uncertainty relation'' between the accuracy and success probability of purification characterized by a regime of $\{\epsilon,p\}$ that is not achievable by any free protocol, as illustrated in Fig.~\ref{fig: forbidden area}. 
In particular, by letting $\epsilon =0$ we directly rule out the possibility of perfect purification:
\begin{corollary}
    It is impossible to exactly transform a full-rank primitive state to a pure target resource state by any free protocol, even probabilistically. 
\end{corollary}

\begin{figure}[]
\begin{center}
\includegraphics{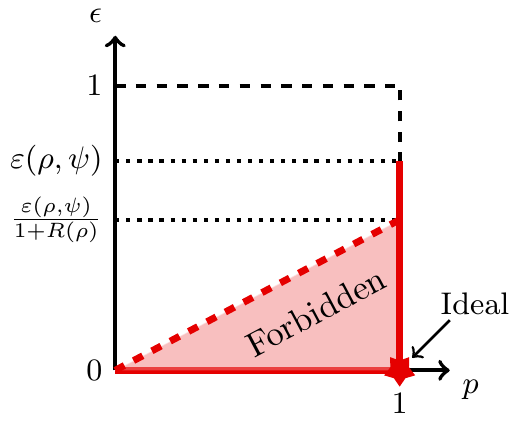}
\end{center}
\caption{Interplay between the transformation error $\epsilon$ and success probability $p$. The lower right corner represents the most ideal scenario where $\epsilon$ is small and $p$ is large. The red region and solid lines represent the forbidden regime such that no purification protocol with the corresponding parameters can exist. $\ve(\rho,\psi) = \lambda_{\min}(\rho)(1-f_\psi)$.
}
\label{fig: forbidden area}
\end{figure}
Below we sketch our approach to proving the above results. See the Supplemental Materials for the detailed proof and extended discussions.
\begin{proof}
(Sketch) The cornerstone of our proof is an information-theoretic quantity called the quantum hypothesis testing relative entropy \cite{buscemi_2010,WangRenner12:hypothesis_testing}, which is defined as $D_H^\epsilon(\rho\|\sigma) := -\log \min \big\{\tr M\sigma\,|\tr \rho M \geq 1-\epsilon,\, 0 \leq M \leq \1\big\}$ for two quantum states $\rho$ and $\sigma$.
The induced resource measure given by $\mathfrak{D}_H^\epsilon(\rho):= \min_{\omega\in\bcF}D_H^\epsilon(\rho\|\omega)$, which was recently related to the rates of certain one-shot resource trading tasks \cite{LBT19}, is shown to exhibit a peculiar property: for any full-rank $\rho$, it vanishes at $\epsilon=0$ and is continuous around it.   The proof then follows from suitably combining this property with the monotonicity of $\mathfrak{D}_H$ (nonincreasing under free operations).   
\end{proof}
Note that Ref.~\cite{marvian2020coherence} reached a similar conclusion for time-translationally invariant operations in coherence theory.
Also note that the full-rank assumption and the error bound can be improved in certain cases by different proof methods, which will be elaborated in follow-up works.    

Remarkably, the noisy primitive state $\rho$ could be much more valuable in terms of other resource measures and tasks or live in much higher dimensions than the pure target state ${\psi}$.  
However, the possibility of trading $\rho$ for ${\psi}$, even probabilistically, is ruled out.  This should be contrasted with the case of pure input $\rho$, where there are no such limitations. 
An illustrative toy example in terms of the theory of coherence is given in Fig.~\ref{fig: coherence example}, where $\rho$ is a slightly noisy version of the maximally coherent state $\ket{+}$ (which can be arbitrarily close to $\ket{+}$), while ${\psi}$ is a pure target state very close to the basis (incoherent) state $\ket{1}$. It is clear from geometrical intuitions that common coherence measures (see e.g.~\cite{coherenceRMP}) assign much greater value to $\rho$ than to ${\psi}$, and it is known that $\ket{+}$ can be transformed to any other state, including ${\psi}$ \cite{PhysRevLett.113.140401,WinterYang16}. However, our results indicate that there is always a neighborhood of ${\psi}$ that cannot be reached starting from $\rho$.
This highlights the special role of $\mathfrak{D}_H$ among all resource measures, and indicates sharp distinctions between pure state transformation problems and mixed state ones.
\begin{figure}[]
\begin{center}
\includegraphics{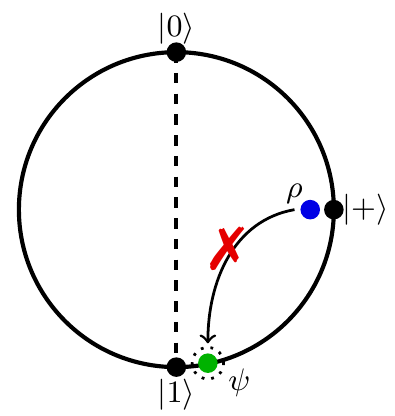}
\end{center}
\caption{A qubit coherence theory example illustrated using the Bloch sphere. Here $\rho$ is a mixed state close to the maximally coherent state $\ket{+}$, and $\psi$ is a pure state close to basis state $\ket{1}$. Our no-go theorems indicate that an arbitrarily accurate probabilistic transformation from $\rho$ to $\psi$ is impossible.}
\label{fig: coherence example}
\end{figure}

The following scheme of resource purification, usually known as ``distillation'' or ``concentration,'' is of the greatest practical interest: one has access to multiple copies of some noisy primitive resource state, and the goal is to ``distill'' certain useful pure resource states to some desired accuracy by free operations while consuming as few copies of the primitive state as possible.  Most notably, the distillation of entanglement \cite{BBPS96:ent_dist,BBPSSW96:ent_dist,BDSW96:ent_qec}, coherence \cite{WinterYang16,Fang2018,Regula2018} and magic states \cite{bravyi2005} has been extensively studied as a key subroutine in quantum communication and computation. 
Therefore, the amount of primitive states needed to accomplish the desired distillation, namely the resource cost or overhead, is a key figure of merit for distillation protocols.  
To present the most general result, we consider error on the entire output state (which could be a collection of unit states) for now.     
As we now show, our no-go theorems indicate fundamental lower bounds on the total overhead of distillation.   
\begin{theorem}\label{thm: total overhead}
    Consider the task of distilling some pure target resource state $\psi$ with error at most $\epsilon$,
    from $n$ copies of primitive state ${\hat\rho}$.  
    For any full-rank ${\hat\rho}$, there does not exist any probabilistic protocol with success probability $p$ that accomplishes the task if the following is not satisfied:
    \begin{equation}
        n \geq \log_{\frac{1+R({\hat\rho})}{\lambda_{\min}({\hat\rho})}}\frac{(1-f_\psi)p}{\epsilon}.    \label{eq:prob_n_main}
    \end{equation}
For deterministic case ($p=1$), the bound can be improved to $n \geq \log_{{1}/{\lambda_{\min}({\hat\rho})}}{(1-f_\psi)}/{\epsilon}$.
\end{theorem}
\begin{proof}
Let $\hat{\rho}^{\otimes n}$ be the primitive state in Theorem~\ref{thm: nogo 1}.  Notice that $\lambda_{\min}(\hat{\rho}^{\otimes n}) = \lambda_{\min}(\hat{\rho})^n$.  For the deterministic case, Theorem~\ref{thm: nogo 1} implies that for any full-rank state $\hat{\rho}$,
we have
\begin{equation}
    \epsilon\geq \lambda_{\min}(\hat{\rho}^{\otimes n})(1-f_\psi) = \lambda_{\min}(\hat{\rho})^n(1-f_\psi). 
\end{equation}
    This directly translates to $n \geq \log_{{1}/{\lambda_{\min}({\hat\rho})}}{(1-f_\psi)}/{\epsilon}$. 
For the probabilistic case, note the following: by the definition of $R(\hat{\rho})$, there exists some state $\tau$ such that $\hat{\rho} + R(\hat{\rho})\tau = (1+R(\hat{\rho}))\omega$ where $\omega\in \cF$. By expanding this equation, we obtain
\begin{equation}
    \omega^{\otimes n} = \frac{1}{(1+R(\hat{\rho}))^n}\hat{\rho}^{\otimes n} + \frac{(1+R(\hat{\rho}))^n-1}{(1+R(\hat{\rho}))^n}\tau',
\end{equation}
where $\omega^{\otimes n}\in\cF$ axiomatically \cite{BrandaoGour15} and $\tau'$ is a density operator.  Therefore, $1+R(\hat{\rho}^{\otimes n}) \leq (1+R(\hat{\rho}))^n$.  Now, by Theorem 1, for any full-rank state $\hat{\rho}'$,
we have
\begin{equation}
    \epsilon/p \geq \frac{\lambda_{\min}(\hat{\rho}'^{\otimes n})(1-f_\psi)}{1+R(\hat{\rho}'^{\otimes n})} \geq \frac{\lambda_{\min}(\hat{\rho}')^n(1-f_\psi)}{(1+R(\hat{\rho}'))^n}.
\end{equation}
This directly translates to Eq.~(\ref{eq:prob_n_main}).
\end{proof}

The above theorem indicates that for distillation protocols that succeed with at least a constant probability (that does not vanish when reducing the target $\epsilon$), the total overhead must scale at least as $\Omega(\log(1/\epsilon))$ as $\epsilon\rightarrow 0$.



As an important application, we discuss magic state distillation~\cite{bravyi2005}, which is a major component of the leading scheme for fault tolerance \cite{NielsenChuang,Shor96,CampbellTerhalVuillot17}. 
Here, the so-called Clifford operations are considered free since they admit fault-tolerant implementations thanks to stabilizer codes \cite{Gottesman96,Gottesman97:thesis,CRSS97,NielsenChuang}, but meanwhile their computational power is very limited---due to the celebrated Gottesman-Knill theorem, they can even be efficiently simulated by classical computers \cite{Gottesman98,AaronsonGottesman04,NielsenChuang}. To achieve universal quantum computation, one needs non-Clifford gates such as $T = \mathrm{diag}(1, e^{i\pi/4})$.   A standard approach is to distill high-quality magic state $\ket{T}=(\ket{0}+e^{i\pi/4}\ket{1})/\sqrt{2}$ from sufficiently many noisy magic states offline, and then use an approximate $\ket{T}$ state to emulate each low-error logical $T$-gate in the circuit via a technique called state injection or gadgetization \cite{GottesmanChuang99:injection}.  
Since the resource cost of this magic state distillation component is dominant in the entire scheme, it is crucial to understand the ultimate limitations to its efficiency.

We now address this problem by tailoring our general results to the practical magic state distillation settings, providing the first rigorous understanding of the resources required for fault-tolerance schemes.
(Note that the resource theory ideas have advanced our understanding of magic states and quantum computation in various other ways \cite{Veitch_2014,HowardCampbell17:magic_rt,WangWilde:magic,CampbellSeddon19,WangWildeSu19:magic_channel,Beverland19:lowerbound}.)
Known protocols for magic state distillation are commonly based on concatenating error correction subroutines using stabilizer codes to probabilistically produce an output with sufficiently high quality upon passing the syndrome measurements.  
The output could take the form of a large state with each marginal sufficiently close to a unit target state, in which case we are also interested in the {average overhead}, i.e.~the total overhead divided by the number of marginals.  
Here we only showcase the $T$-state result, but the bounds for other useful magic states (see e.g.~\cite{Haah2018codesprotocols}) can be similarly obtained by plugging in corresponding parameters.



\begin{theorem}\label{thm: magic overhead}
    Consider the following general form of magic state distillation task: given $n$ copies of full-rank primitive magic states ${\hat\rho}$, output an $m$-qubit state $\tau$ such that $\tr \tau_i T = \bra{T}\tau_i\ket{T} \geq 1-\epsilon, \forall i = 1,...,m$ where $\tau_i = \tr_{\overline{i}}\tau$ is the $i$-th qubit.  
    Then 
    the average overhead of any free probabilistic protocol that succeeds with probability $p$ must obey
    \begin{equation}\label{eq:magic overhead}
        C := n/m \geq \frac{1}{m}\log_{\frac{1+R({\hat\rho})}{\lambda_{\min}({\hat\rho})}}   \frac{((4-2\sqrt{2})^m -1)p}{(4-2\sqrt{2})^m m \epsilon}.
    \end{equation}
\end{theorem}
\begin{proof}
By applying the union bound, we have $\bra{T^{\otimes m}}\tau\ket{T^{\otimes m}} \geq 1-m\epsilon$. Also notice that $f_{T^{\otimes m}} = (4-2\sqrt{2})^{-m}$ \cite{Campbell11:catalysis,BravyiGosset16,Bravyi18:lowrank,LBT19}.  By plugging everything into Eq.~({\ref{eq:prob_n_main}}) we obtain the claimed bound.
\end{proof}

In the analyses of magic state distillation protocols, one is particularly interested in the exponent $\gamma$ in the asymptotic average overhead $O(\log^\gamma(1/\epsilon))$ as $\epsilon\rightarrow 0$.
A subtlety of our lower bound is that the output size $m$ could depend on the target $\epsilon$ for specific protocols.   Thus, to understand the scaling, one needs to take into account the behavior of $m$ as well.     
There are two key implications of our bound to code-based distillation protocols.  Assuming nonvanishing success probability (the passing probability of deeper rounds of concatenation converges sufficiently fast to one), we conclude the following: 
(i) It is impossible to construct a protocol with sublogarithmic average overhead ($\gamma < 1$) with any $[n,k,d]$ code such that $k\leq d$. This can be seen by plugging $m = k^\nu$ and $\log(1/\epsilon) \sim d^\nu$ into Eq.~(\ref{eq:magic overhead}). 
This in particular implies a $\gamma\geq 1$ bound for $k=1$ codes, in response to open questions raised in e.g.~\cite{BravyiHaah12,HastingsHaah18}. Note that the best known such codes allow $\gamma \rightarrow 2$ \cite{Haah2017magicstate,Haah2018codesprotocols}, so there is still potential room for improvement.
(ii) Any $\gamma < 1$ protocol must have a scale (size of the output) that diverges under concatenation.   It was actually believed that no codes allowing $\gamma < 1$ exist \cite{BravyiHaah12}, but the recent breakthrough work by Hastings and Haah \cite{HastingsHaah18} gives a peculiar example of such a code (see also \cite{KrishnaTillich18}), prompting the question of whether there is any fundamental limit.
(There, indeed, the codes employed have $k>d$.)  Our results indicate that, although the average overhead of such a protocol is considered low, its output size must grow rapidly as we reduce $\epsilon$, which inevitably blows up the overall cost. 
Finally, we make a basic extension to the channel resource theory setting (see e.g.~\cite{LiuWinter2018,Gour18comparison,LiuYuan:channel}), a more general setting of surging interest recently, which directly applies to quantum channels, gates, and dynamical processes, etc.
We show that, under the analogous golden rule, it is generally impossible to perfectly transform a noisy quantum channel into a unitary resource channel.  
A straightforward implication of this result is that the zero-error quantum capacity of generic noisy channels, e.g.~the depolarizing channel, is zero. 
 See the Supplemental Materials for detailed statements and proofs. More comprehensive studies of the channel setting will be left for follow-up.




To conclude, this work establishes quantitative bounds on the accuracy and efficiency of purifying noisy quantum resources and thus draws practical boundaries for quantum error correction and mitigation, by employing one-shot quantum resource theory techniques.
Our results universally apply to quantum resources of any reasonable kind.  The bounds depend only on very few parameters that concisely encode relevant properties of the noise, the target state, and the resource theory, and are thus easy to analyze.
Like the no-cloning theorem, our ``no-purification'' theorems stem from fundamental laws of quantum mechanics at bottom. 
We demonstrate the power and practical relevance of our general methods by establishing strong lower bounds on the overhead of distillation tasks (e.g.~magic state distillation), which provide rigorous understandings of and useful benchmarks for the resource requirements of practical quantum technologies, in particular fault-tolerant quantum computation, as the Heisenberg limit did for quantum metrology. 

An important future work is to investigate to what extent our various bounds can be approached, both by general means and in specific theories. For instance, it remains to be checked how close the state-of-the-art entanglement purification protocols (see e.g.~\cite{Krastanov2019optimized}) are to the fundamental limits set here. We also expect our general, primary results to see improvements in various cases and more generally, stimulate further studies on optimal quantum resource purification.  It would also be interesting to further understand the approximate and probabilistic regimes of unitary channel simulation, due to its connections to the fields of quantum Shannon theory, gate and circuit synthesis etc.  In sum, a key message of this work is that the cost of practically implementing quantum technologies 
or experiments could not be indefinitely improved in general, due to noise effects.  As we are now witnessing an exciting paradigm shift from blueprinting quantum advantages in theory to actually putting them into practice \cite{Preskill2018quantumcomputingin,supremacy19},  we anticipate that such rigorous understandings of the fundamental obstacles will serve as an important guideline and have far-reaching implications for quantum science and technology.    


\smallskip
\emph{ Acknowledgments.} We thank Earl Campbell, Daniel Gottesman, Gilad Gour, Jeongwan Haah, Sirui Lu, Bartosz Regula, Ryuji Takagi, Julio I.~de Vicente, Andreas Winter for helpful discussions and feedback, and the anonymous referees for several valuable comments.  K.F. is supported by the University of Cambridge Isaac Newton
Trust Early Career Grant No.~RG74916.
Z.-W.L. is supported by Perimeter Institute for Theoretical Physics.
Research at Perimeter Institute is supported by the Government
of Canada through Industry Canada and by the Province of Ontario through the Ministry
of Research and Innovation.

\bibliography{nogo}


\clearpage

\onecolumngrid
\begin{center}
\vspace*{.5\baselineskip}
{\textbf{\large Supplemental Materials: No-Go Theorems for Quantum Resource Purification}}\\[1pt] \quad \\
\end{center}

\renewcommand{\theequation}{S\arabic{equation}}
\renewcommand{\thetheorem}{S\arabic{theorem}}
\setcounter{equation}{0}
\setcounter{figure}{0}
\setcounter{table}{0}
\setcounter{section}{0}

In the Supplemental Materials, we provide more detailed proofs and discussions of several results in the main text. We may reiterate some
of the steps to make the Supplemental Materials more explicit and self-contained.

\section{Technical Lemmas}

Recall the definition of quantum hypothesis testing relative entropy as
\begin{align}
  D_H^\epsilon(\rho\|\sigma) := -\log \beta_\epsilon(\rho\|\sigma),\quad \text{with}\quad \beta_\epsilon(\rho\|\sigma) := \min \big\{\tr M\sigma\,|\tr \rho M \geq 1-\epsilon,\, 0 \leq M \leq \1\big\}.
\end{align}

The following technical lemmas will be used in the proofs of our main results.

\begingroup
\renewcommand{\theproposition}{S1}
\begin{lemma}\label{lem: continuity bound}
For any full rank states $\rho$ and any quantum state $\sigma$, their quantum hypothesis testing relative entropy is continuous around $\epsilon = 0$. That is, for any $0 \leq \epsilon < \lambda_{\min}(\rho)$ where $\lambda_{\min}(\rho)$ is the smallest eigenvalue of $\rho$, it holds that
\begin{align}
 0 \leq D_H^\epsilon(\rho\|\sigma)\leq \log \frac{\lambda_{\min}(\rho)}{\lambda_{\min}(\rho) - \epsilon}.
\end{align}
\end{lemma}
\endgroup
\begin{proof}
Suppose $M$ is an optimal measurement operator that attains $D_H^\epsilon(\rho\|\sigma)$. Then we have $0\leq M \leq \1$, $\tr \rho M \geq 1-\epsilon$ and $D_H^\epsilon(\rho\|\sigma) = -\log \tr M \sigma$. Denote the non-zero eigenvalues of $\rho$ and $M$ as $\{\lambda_i\}_{i=1}^d$ and $\{m_j\}_{j=1}^k$, which are both sorted in an non-increasing order. Let $\lambda_{\min} = \min_{i} \lambda_i$ and $m_{\min} = \min_{j} m_j$. Since $\rho$ is full rank, we have $k\leq d$. We first argue that $k = d$, i.e., $M$ is full rank. Suppose $k < d$, then we have 
\begin{align}
1-\epsilon \leq \tr \rho M \leq \sum_{i=1}^k \lambda_i m_i \leq \sum_{i=1}^{k}\lambda_i \leq 1 - \lambda_{\min},\label{eq: Lemma proof tmp1}
\end{align}
where the second inequality follows from the von Neumann’s trace theorem~\cite[Theorem 7.4.1.1]{Horn2012}, and the third inequality follows since $m_i \leq 1$. Therefore, Eq.~\eqref{eq: Lemma proof tmp1} contradicts to the assumption that $\epsilon < \lambda_{\min}$ and thus $k = d$.  Similar to Eq.~\eqref{eq: Lemma proof tmp1}, we have
\begin{align}
  1-\epsilon \leq \tr \rho M \leq \sum_{i=1}^d \lambda_i m_i \leq \left(\sum_{i=1}^{d-1}\lambda_i\right) + \lambda_{\min}m_{\min},\label{eq: Lemma proof tmp2}
\end{align}
where the last inequality follows from the rearrangement inequality and the fact that $m_i\leq 1$.
This implies $m_{\min} \geq 1-\epsilon/\lambda_{\min}$. Then we have $M \geq m_{\min} \1 \geq (1-\epsilon/\lambda_{\min}) \1$, and thus 
\begin{align}
  D_H^\epsilon(\rho\|\sigma) = -\log \tr M \sigma \leq -\log \tr (1-\epsilon/\lambda_{\min}) \sigma = \log \frac{\lambda_{\min}(\rho)}{\lambda_{\min}(\rho) - \epsilon}.
\end{align}
This completes the proof.
\end{proof}

\begingroup
\renewcommand{\theproposition}{S2}
\begin{lemma}\label{lem: flagged state DH error}
	 For any two flagged quantum states 
  $\rho_i  = p_i \ket{0}\bra{0}\ox \sigma_i + (1-p_i) \ket{1}\bra{1}\ox \omega_i$ with $i\in \{1,2\}$ and $p_i \in [0,1]$,  it holds that
$\beta_\epsilon(\rho_1\|\rho_2) \leq p_2 \beta_\epsilon(\sigma_1\|\sigma_2) + (1-p_2)\beta_\epsilon(\omega_1\|\omega_2)$.
\end{lemma}
\endgroup
\begin{proof}
Suppose $\beta_\epsilon(\sigma_1\|\sigma_2)$ and $\beta_\epsilon(\omega_1\|\omega_2)$ are achieved by optimal measurement operators $M$ and $N$ respectively. Then we can take $Q = \ket{0}\bra{0} \ox M + \ket{1}\bra{1} \ox N$. It is clear that $0 \leq Q \leq \1$ and $\tr Q \rho_1 = p_1\tr M \rho + (1-p_1) \tr N \omega_1 \geq p_1(1-\epsilon) + (1-p_1) (1-\epsilon) = 1-\epsilon$. So $Q$ is a feasible measurement operator for $\beta_\epsilon(\rho_1\|\rho_2)$. Thus we have
\begin{align}
    \beta_\epsilon(\rho_1\|\rho_2) &\leq \tr Q \rho_2  = p_2 \tr M \sigma_2 + (1-p_2) \tr N \omega_2 = p_2 \beta_\epsilon(\sigma_1\|\sigma_2) + (1-p_2) \beta_\epsilon(\omega_1\|\omega_2),
\end{align}
which completes the proof.
\end{proof}

\begingroup
\renewcommand{\theproposition}{S3}
\begin{lemma}\label{lem: robustness}
   For any linear suboperation $\cL$, there exists a free state $\omega\in\cF$ such that
    $\tr \cL(\omega) \geq (1+R(\rho))^{-1} \tr \cL(\rho)$, where  $R(\rho)\equiv \min\{s| \exists\, \sigma, s \geq 0, \text{ s.t. } (\rho+ s\sigma)/(1+s) \in \bcF\}$ is the generalized robustness of state $\rho$.
\end{lemma}
\endgroup
\begin{proof}
    By definition of $R(\rho)$, there exists $\o\in\cF$ such that $\omega = \frac{1}{1+R(\rho)}\rho + \frac{R(\rho)}{1+R(\rho)}\sigma$.  By linearity of $\cL$,
    \begin{equation}
       \tr \cL(\omega) = \frac{1}{1+R(\rho)}\tr \cL(\rho) + \frac{R(\rho)}{1+R(\rho)}\tr \cL(\sigma),
    \end{equation}
    and thus the desired bound directly follows.
\end{proof}

\section{No-go theorems for purification}

\noindent \textbf{Restatement of Theorem 1.}
Given any full-rank primitive state $\rho \not\in \bcF$ and any pure target resource state $\psi \not\in \bcF$, the following relation between the  success probability $p$ and transformation error $\epsilon$ must hold for any free probabilistic protocol:
\begin{equation}
\frac{\epsilon}{p} \geq  \frac{\lambda_{\min}(\rho)(1-f_\psi)}{1+R(\rho)}. \label{eq:prob_purif}
\end{equation}
where $\lambda_{\min}(\rho)$ is the smallest eigenvalue of $\rho$, $f_\psi:= \max_{\o\in \bcF} \tr (\ket{\psi}\bra{\psi} \o)$ is the maximum overlap between $\psi$ and free states $\bcF$, and $R(\rho):= \min\{s| \exists\, \sigma, s \geq 0, \text{ s.t. } (\rho+ s\sigma)/(1+s) \in \bcF\}$ is the generalized robustness of state $\rho$.  For the deterministic case ($p=1$), the bound can be improved to 
\begin{equation}
 \epsilon \geq \lambda_{\min}(\rho)(1-f_\psi).     \label{eq:determ_purif}
\end{equation}

\begin{proof}
(Deterministic case.)
Suppose there is a free operation $\cE$ such that $\cE(\rho) = \sigma$ and  $\tr \psi \sigma \geq 1-\epsilon$ with $\epsilon < \ve(\rho,\psi) \equiv \lambda_{\min}(\rho)(1-f_\psi)$. Consider the quantum test $\{\psi,\1-\psi\}$, we have $D_H^\epsilon(\sigma\|\omega) \geq -\log \tr \psi\omega$ for all $\omega \in \bcF$. Then the following chain of inequalities holds
\begin{align}
\hspace{-0.2cm} -\log f_\psi &  \leq \min_{\o\in\bcF}D_{H}^\epsilon(\sigma\|\o)\\
& = \min_{\o\in\bcF}D_{H}^\epsilon(\cE(\rho)\|\o)\\
& \leq \min_{\o\in\boldsymbol \cF}D_{H}^\epsilon(\cE(\rho)\|\cE(\o))\\
& \leq \min_{\o\in\boldsymbol \cF}D_{H}^\epsilon(\rho\|\o)\\
& \leq  \log \left({\lambda_{\min}(\rho)}/({\lambda_{\min}(\rho) - \epsilon})\right)
\end{align}
where the second line follows from the assumption that $\cE(\rho) = \sigma$, the third line follows since it is minimized over a smaller set due to $\cE(\omega) \in \bcF$, $\forall \omega\in \bcF$, the fourth line follows from the data-processing inequality of $D_H^\epsilon$ \cite{WangRenner12:hypothesis_testing}, the last line follows from the continuity bound in Lemma~\ref{lem: continuity bound} (applicability guaranteed by the assumption $\epsilon <  \ve(\rho,\psi)$). A simple calculation gives us $\varepsilon \geq \lambda_{\min}(\rho)(1-f_\psi)$, which forms a contradiction with our assumption.

(Probabilistic case.)
Suppose there is a probabilistic protocol $\cE_{A\to FB} (\rho_A) = \ket{0}\bra{0}_F\ox \cL_{A\to B}(\rho_A) + \ket{1}\bra{1}_F \ox \cG_{A\to B}(\rho_A)$ such that $\cL_{A\to B}(\rho_A) = p \sigma_B$ and $\tr \sigma \psi \geq 1-\epsilon$ with $\epsilon < p(1+R(\rho))^{-1}\ve(\rho,\psi)$. Based on Lemma~\ref{lem: robustness}, let us consider a free state $\omega_1$ such that $\tr \cL( \omega_1) \geq (1+R(\rho))^{-1} \tr \cL(\rho)$ and take $\omega_2 = \cL(\o_1)/\tr \cL(\o_1)$. Then we have 
\begin{align}\label{eq: nogo 2 tmp1}
f_\psi = \max_{\o \in \bcF} \tr \omega \psi \geq \tr  \omega_2 \psi \geq \beta_\epsilon(\sigma\|\o_2),
\end{align}
where the last inequality follows by considering the quantum test $\{\psi, \1-\psi\}$. On the other hand, we have
\begin{align}
	 \beta_\epsilon(\sigma\|\omega_2) & = \beta_\epsilon\left(\frac{\cL(\rho)}{\tr \cL(\rho)}\Big\|\frac{\cL(\o_1)}{\tr \cL(\o_1)}\right)\\
	 & \geq \frac{\beta_\epsilon(\cE(\rho)\|\cE(\o_1)) - (1-\tr \cL(\o_1))}{\tr \cL(\o_1)}\\
	 & \geq \frac{\beta_\epsilon(\rho\|\o_1) - (1-\tr \cL(\o_1))}{\tr \cL(\o_1)}\\
	 & \geq \frac{(\lambda_{\min}(\rho)-\epsilon)/\lambda_{\min}(\rho) - (1-\tr \cL(\o_1))}{\tr \cL(\o_1)},\label{eq: nogo 2 tmp2}
\end{align}
where the second line follows from Lemma~\ref{lem: flagged state DH error}, the third line follows from the data-processing inequality, and the last line follows from the continuity bound in Lemma~\ref{lem: continuity bound} (applicability guaranteed by the assumption $\epsilon < p(1+R(\rho))^{-1}\ve(\rho,\psi) \leq \ve(\rho,\psi)$) and the assumption that $\rho$ is full rank.
Combining~\eqref{eq: nogo 2 tmp2} with~\eqref{eq: nogo 2 tmp1}, we have $\epsilon \geq \tr \cL(\o_1)\ve(\rho,\psi)$. Recall that $\o_1$ is defined as a free state such that $\tr \cL( \omega_1) \geq (1+R(\rho))^{-1} \tr \cL(\rho) = (1+R(\rho))^{-1} p$, we obtain $\epsilon \geq p(1+R(\rho))^{-1}\ve(\rho,\psi)$, which forms a contradiction with our assumption. 
\end{proof}

\begin{remark}
     A slightly weaker version of Eq.~(\ref{eq:determ_purif}) is recovered by letting $p=1$.
It is also possible to get rid of the $(1+R(\rho))^{-1}$ factor and obtain a stronger bound that covers Eq.~(\ref{eq:determ_purif}) under certain restrictions.  For example, suppose the theory admits a resource destroying channel \cite{LiuHuLloyd17:rdmap} $\Lambda$, and the allowed free suboperations are those commuting with $\Lambda$ (such as dephasing-covariant incoherent suboperations for coherence theory~\cite{Fang2018}).  Then for any free suboperation $\cL$, it holds that $\tr \cL(\rho) = \tr \Lambda\circ \cL(\rho) = \tr \cL \circ \Lambda (\rho)$, which indicates that there always exists a free state $\omega = \Lambda (\rho)$ such that $\tr\cL(\o) = p$, and therefore the bound reduces to $\epsilon/p \geq \ve(\rho,\psi)$.
\end{remark}

\section{No-go theorem for unitary channel simulation}

Analogous to the resource theories of quantum states, a resource theory of quantum channels can also be built upon two basic ingredients: the set of free channels $\bcO$ and the set of free superchannels (maps from channels to channels) $\bTheta$, with the golden rule $\bTheta(\bcO) \subseteq \bcO$ which selects all superchannels that can be possibly allowed (including those admitting implementations by free combs considered in \cite{LiuWinter2018}).
Now the general question becomes whether there exists a free superchannel $\Pi \in \bTheta$ that maps one quantum channel $\cN$ to another quantum channel $\cM$, i.e.,
$\Pi(\cN_{A\to B}) = \cM_{C\to D}$.      If so, we say that $\cM$ can be simulated by $\cN$.

In analogy to the state distillation tasks where one aims to turn a noisy state into a pure one, here we want to turn a noisy channel into a unitary one, which preserves information.  Below we show an elementary channel version of the no-purification theorems, which says that perfect simulation of unitary channels are generically impossible. 

We say a quantum channel $\cN$ has free component if there exists free channel $\cE \in \bcO$ and another quantum channel $\cM$ such that $\cN = p \cE + (1-p) \cM$ with $p > 0$.
We also need the definition of the channel's min-relative entropy and its monotonicity under superchannels.
The channel's min-relative entropy is defined as 
\begin{align}
  D_{\min}(\cN\|\cM) & \equiv \sup_{\rho_{AA'}} D_{\min}(\cI_{A} \ox \cN_{A'\to B}(\rho_{AA'})\|\cI_A \ox \cM_{A'\to B}(\rho_{AA'}))
\end{align}
where $D_{\min}(\rho\|\sigma)\equiv D_H^0(\rho\|\sigma)$, the supremum is taken over all quantum states on systems $AA'$ and $\cI$ is the identity map. The monotonicity of the channel's min-relative entropy holds as follows:

\begingroup
\renewcommand{\theproposition}{S4}
\begin{lemma}\label{Dmin channel DPI}
  For any superchannel $\Pi$ and quantum channels $\cN_{A'\to B}$ and $\cM_{A'\to B}$, it holds $D_{\min}(\Pi(\cN)\|\Pi(\cM)) \leq D_{\min}(\cN\|\cM)$.
\end{lemma}
\endgroup
\begin{proof}
Note that any superchannel $\Pi$ can be implemented by pre- and post- quantum processings as $\Pi(\cN_{A\to B}) = \cD_{BR\to D} \circ (\cN_{A\to B}\ox \id_{R\to R})\circ \cE_{C \to AR}$, where $\cE,\cD$ are channels \cite{Chiribella_2008}.
 Suppose $\Pi$ is realized by $\{\cE,\cD\}$ and the optimal solution of $D_{\min}(\Pi(\cN)\|\Pi(\cM))$ is achieved by state $\rho_{CE}$. Denote $\sigma_{ARE} = \cE_{C \to AR} (\rho_{CE})$. Then we have the following chain of inequalities,
  \begin{align}
  & D_{\min}(\Pi(\cN)\|\Pi(\cM))\notag\\
  & = D_{\min}(\cD_{BR\to D} \circ \cN_{A\to B} (\sigma_{ARE}) \|\cD_{BR\to D} \circ \cM_{A\to B}(\sigma_{ARE}))\\
  & \leq D_{\min}(\cN_{A\to B}(\sigma_{ARE})\|\cM_{A\to B}(\sigma_{ARE}))\\
  & \leq D_{\min}(\cN\|\cM),
  \end{align}
  where the first inequality follows from the data-processing inequality of $D_{\min}$, and the second inequality follows by definition.
\end{proof}

\begingroup
\renewcommand{\theproposition}{S1}
\begin{theorem}[No-go theorem for channel simulation]\label{thm: channel nogo 1}
 Given any primitive channel $\cN_{A\to B} \not\in \boldsymbol\cO(A\to B)$ which has free component, and any target unitary resource channel $\cU_{C\to D} \not\in \boldsymbol\cO(C\to D)$, there is no free superchannel $\Pi$ transforming $\cN_{A\to B}$ to $\cU_{C\to D}$.
\end{theorem}
\endgroup

\begin{proof}
We prove this by contradiction. Suppose there is a free superchannel $\Pi$ such that $\Pi(\cN) = \cU$. Denote $J_{\cU} = \cI_R\ox \cU_A(\Phi_{RA})$ as the Choi state of $\cU$ and $\Phi_{RA} = \frac{1}{|A|} \sum_{i,j=1}^{|A|} \ket{ii}\bra{jj}$ is the maximally entangled state. Since $\cU$ is a unitary channel, we know that $J_{\cU}$ is a pure state. By the assumption $\cU \not\in \bcO$, for any $\cE \in \bcO$ it holds that $D_{\min}(\cU\|\cE) \geq D_{\min}(J_{\cU}\|J_{\cE}) = -\log \tr J_{\cU} J_{\cE} > 0$. On the other hand, suppose $\cN = p\cE + (1-p)\cM$ with $\cE \in \bcO$, we have $D_{\min}(\cN\|\cE) = D_{\min}(p\cE + (1-p)\cM\|\cE) = 0$. Then we have
\begin{align}
    0 < \min_{\cE \in \bcO} D_{\min}(\cU\|\cE) \leq \min_{\cE \in \bcO} D_{\min}(\cU\|\Pi(\cE)) =  \min_{\cE \in \bcO} D_{\min}(\Pi(\cN)\|\Pi(\cE)) \leq \min_{\cE \in \bcO} D_{\min}(\cN\|\cE) = 0,
\end{align}
where the second inequality follows since the minimization is restrict to $\cE \in \Pi(\boldsymbol \cO)$  on the r.h.s., and the third inequality follows from the monotonicity property, Lemma~\ref{Dmin channel DPI}. This forms a contradiction. 
\end{proof}
The optimal rate of simulating a noiseless quantum channel corresponds to the well-studied quantum capacity (see e.g.~\cite{fang2018quantum,Fang2020simulation}). 
As a result, the above theorem implies that the zero-error quantum capacity of a quantum channel with free component, such as the quantum depolarizing channel, is zero.

\end{document}